\documentclass[preprint]{sig-alternate-05-2015}

\usepackage{enumerate}
\usepackage[linkcolor=blue,colorlinks=true,citecolor=blue,urlcolor=blue]{hyperref} 

\begin{document}

\toappear{ISSAC'16, July 19--22, 2016, Waterloo, ON, Canada. \\
  ACM ISBN.
  DOI: http://dx.doi.org/10.1145/2930889.2930928}

\CopyrightYear{2016}
\setcopyright{acmlicensed}
\conferenceinfo{ISSAC'16,}{July 19--22, 2016, Waterloo, ON, Canada}
\isbn{978-1-4503-4380-0/16/07}\acmPrice{\$15.00}
\doi{http://dx.doi.org/10.1145/2930889.2930928}

\title{Fast Computation of Minimal Interpolation Bases \\ in Popov Form for Arbitrary Shifts}
\numberofauthors{2}
\author{
\alignauthor
Claude-Pierre Jeannerod \\
  \affaddr{Inria, Universit\'e de Lyon \\
  Laboratoire LIP (CNRS, Inria, ENSL, UCBL) \\
  claude-pierre.jeannerod@ens-lyon.fr}
\alignauthor
  Vincent Neiger \\
  \affaddr{ENS de Lyon, Universit\'e de Lyon \\
  Laboratoire LIP (CNRS, Inria, ENSL, UCBL) \\
  vincent.neiger@ens-lyon.fr}
\and
\alignauthor
\'Eric Schost \\
  \affaddr{University of Waterloo \\
  David R. Cheriton School of Computer Science \\
  eschost@uwaterloo.ca}
\alignauthor
Gilles Villard\\
  \affaddr{CNRS, Universit\'e de Lyon \\
  Laboratoire LIP (CNRS, Inria, ENSL, UCBL) \\
  gilles.villard@ens-lyon.fr}
}

\renewcommand\leq\leqslant
\newcommand{\storeArg}{}
\newcounter{notationCounter}
\newcommand{\bigO}[1]{\mathcal{O}(#1)} 
\newcommand{\bigOPar}[1]{\mathcal{O}\left(#1\right)} 
\newcommand{\softO}[1]{\mathcal{O}\tilde{~}(#1)} 
\newcommand{\polmultime}[1]{\mathsf{M}(#1)}
\newcommand{\polmultimePar}[1]{\mathsf{M}{#1}}
\newcommand{\polmatmultime}[1]{\mathsf{MM}(#1)}
\newcommand{\polmatmultimePrime}[1]{\mathsf{MM'}(#1)}
\newcommand{\polmatmultimeBis}[1]{\mathsf{MM''}(#1)}
\newcommand{\polmatmultimePrimeDnc}[1]{\mathsf{\overline{MM'}}(#1)}
\newcommand{\polmatmultimeBisDnc}[1]{\mathsf{\overline{MM''}}(#1)}
\newcommand{\expmatmul}{\omega} 
\newcommand{\algoname}[1]{{\normalfont\textsc{#1}}}
\renewcommand{\ge}{\geqslant} 
\renewcommand{\le}{\leqslant} 
\newcommand{\bzero}{\mathbf{0}} 
\newcommand{\Zint}[1]{\{0,\ldots,#1-1\}} 
\newcommand{\ZZ}{\mathbb{Z}} 
\newcommand{\NN}{\mathbb{N}} 
\newcommand{\ZZp}{\mathbb{Z}_{> 0}} 
\newcommand{\NNp}{\mathbb{N}_{> 0}} 
\newcommand{\var}{X} 
\newcommand{\field}{\mathbb{K}} 
\newcommand{\polSpace}{\mathfrak{F}} 
\newcommand{\polRing}{\field[\var]} 
\newcommand{\matSpace}[1][\rdim]{\renewcommand\storeArg{#1}\matSpaceAux} 
\newcommand{\polMatSpace}[1][\rdim]{\renewcommand\storeArg{#1}\polMatSpaceAux} 
\newcommand{\matSpaceAux}[1][\storeArg]{\field^{\storeArg \times #1}} 
\newcommand{\polMatSpaceAux}[1][\storeArg]{\polRing^{\storeArg \times #1}} 
\newcommand{\row}[1]{\mathbf{\MakeLowercase{#1}}} 
\newcommand{\col}[1]{\mathbf{\MakeLowercase{#1}}} 
\newcommand{\mat}[1]{\mathbf{\MakeUppercase{#1}}} 
\newcommand{\matCoeff}[1]{\MakeLowercase{#1}} 
\newcommand{\vecc}[1]{\mathbf{#1}} 
\newcommand{\trsp}[1]{#1^\mathsf{T}} 
\newcommand{\rdim}{m} 
\newcommand{\cdim}{n} 
\newcommand{\spol}[1][1]{\setcounter{notationCounter}{#1}\addtocounter{notationCounter}{15} \Alph{notationCounter}} 
\newcommand{\smat}[1][1]{\setcounter{notationCounter}{#1}\mat{\Alph{notationCounter}}} 
\newcommand{\spolmat}[1][1]{\setcounter{notationCounter}{#1}\mathbf{\Alph{notationCounter}}} 
\newcommand{\matrow}[2]{{#1}_{#2,*}} 
\newcommand{\matcol}[2]{{#1}_{*,#2}} 
\newcommand{\idMat}[1][\rdim]{\mat{I}_{#1}} 
\newcommand{\shiftMat}[1]{\mat{\var}^{#1\,}} 
\newcommand{\leadingMat}[2][\unishift]{\mathrm{lm}_{#1}(#2)} 
\newcommand{\shift}[2][s]{#1_{#2}} 
\newcommand{\shifts}[1][s]{\mathbf{#1}} 
\newcommand{\sshifts}[1][\shifts]{|#1|} 
\newcommand{\shiftSpace}[1][\rdim]{\ZZ^{#1}} 
\newcommand{\unishift}{\mathbf{0}} 
\newcommand{\rdeg}[2][]{\mathrm{rdeg}_{{#1}}(#2)} 
\newcommand{\order}{\sigma} 
\newcommand{\linFunc}{\ell} 
\newcommand{\linFuncs}{\mathcal{L}} 
\renewcommand{\int}{\row{p}} 
\newcommand{\intBasis}{\mat{P}} 
\newcommand{\intSpace}[1][\rdim]{\polMatSpace[#1][#1]} 
\newcommand{\evMat}{\mat{E}} 
\newcommand{\evRow}{\row{e}} 
\newcommand{\evSpace}[2][\rdim]{\matSpace[#1][#2]}
\newcommand{\mul}{\cdot} 
\newcommand{\mulmat}{\jordan} 
\newcommand{\jordan}{\mat{J}} 
\newcommand{\mulshift}{\mat{Z}} 
\newcommand{\ord}[1]{\mathrm{ord}{(#1)}} 
\newcommand{\nbbl}{n} 
\newcommand{\szbl}{\order} 
\newcommand{\evpt}{x} 
\newcommand{\module}{\mathfrak{M}}
\newcommand{\kernel}{\mathfrak{K}}
\newcommand{\intMod}[1][\evMat,\mulmat]{\mathfrak{I}(#1)}
\newcommand{\rowvec}[1][1]{\setcounter{notationCounter}{#1}\addtocounter{notationCounter}{21} \row{\alph{notationCounter}}} 
\newcommand{\krylov}[2][]{\mathcal{K}_{#1}(#2)} 
\newcommand{\linPolMat}[2][]{\mathcal{E}_{#1}(#2)} 
\newcommand{\polFromLin}[2][]{\mathcal{C}_{#1}(#2)} 
\newcommand{\pivmat}[2][]{\mathcal{P}_{#1}(#2)} 
\newcommand{\tgtmat}[2][]{\mathcal{T}_{#1}(#2)} 
\newcommand{\relmat}[2][]{\mathcal{R}_{#1}(#2)} 
\newcommand{\prioPerm}[1][\shifts]{\mat{\pi}_{#1}}
\newcommand{\prioEval}{\psi_{\shifts}}
\newcommand{\prioIndex}[1][\shifts]{\phi_{#1}}
\newcommand{\minDeg}{\delta}
\newcommand{\minDegs}{\boldsymbol{\delta}}
\newcommand{\rep}[2]{\mathrm{rep}(#1,#2)}
\newcommand{\evMatF}{\mat{F}} 
\newcommand{\evRowF}{\row{f}} 
\newcommand{\supp}{\mu}
\newcommand{\nbpt}{p}
\newcommand{\bY}{ {\boldsymbol{Y}} }
\newcommand{\by}{ {\boldsymbol{y}} }
\newcommand{\bj}{ {\boldsymbol{j}} }
\newcommand{\be}{ {\boldsymbol{e}} }
\newcommand{\bi}{ {\boldsymbol{i}} }
\newcommand{\sbj}{|\bj|}
\newcommand{\expSet}{\Gamma}
\newcommand{\nvars}{r}
\newcommand{\points}{\mathcal{P}}
\newcommand{\supps}{\mathcal{M}}
\newcommand{\sumVec}[1]{|#1|} 
\newtheorem{pbm}{Problem}
\newtheorem{algo}{Algorithm}
\newtheorem{dfn}{Definition}[section]
\newtheorem{thm}[dfn]{Theorem}
\newtheorem{cor}[dfn]{Corollary}
\newtheorem{prop}[dfn]{Proposition}
\newtheorem{lem}[dfn]{Lemma}

\newcommand{\expandMat}{\mathcal{E}}
\newcommand{\quoExp}{\alpha}
\newcommand{\remExp}{\beta}
\newcommand{\degExp}{\delta}
\newcommand{\expand}[1]{\overline{#1}}

\maketitle

\begin{abstract}
We compute minimal bases of solutions for a general interpolation problem,
which encompasses Hermite-Pad\'e approximation and constrained multivariate
interpolation, and has applications in coding theory and security.

This problem asks to find univariate polynomial relations between $\rdim$
vectors of size $\order$; these relations should have small degree with respect
to an input degree shift.
For an arbitrary shift, we propose an algorithm for the computation of an
interpolation basis in shifted Popov normal form with a cost of
$\softO{\rdim^{\expmatmul-1} \order }$ field operations, where $\expmatmul$ is
the exponent of matrix multiplication and the notation $\softO{\cdot}$
indicates that logarithmic terms are omitted.

Earlier works, in the case of Hermite-Pad\'e approximation~\cite{ZhoLab12} and
in the general interpolation case~\cite{JeNeScVi15}, compute non-normalized
bases. Since for arbitrary shifts such bases may have size $\Theta(\rdim^2
\order)$, the cost bound $\softO{\rdim^{\expmatmul-1} \order}$ was feasible
only with restrictive assumptions on the shift that ensure small output sizes.
The question of handling arbitrary shifts with the same complexity bound was
left open.

To obtain the target cost for any shift, we strengthen the properties of the
output bases, and of those obtained during the course of the algorithm: all the
bases are computed in shifted Popov form, whose size is always $\bigO{\rdim
\order}$. Then, we design a divide-and-conquer scheme. We recursively reduce
the initial interpolation problem to sub-problems with more convenient shifts
by first computing  information on the degrees of the intermediate bases.  
\end{abstract}

%
%


\vspace{0.3cm}

\section{Introduction}
\label{sec:intro}

\vspace{-0.2cm}

\subsection{Problem and main result}
\label{subsec:problem_results}

We focus on the following interpolation problem
from~\cite{BarBul92,BecLab00}. For a field $\field$ and some positive integer
$\sigma$, we have as input $\rdim$ vectors $\row{e}_1,\ldots,\row{e}_\rdim$ in
$\field^{1\times \order}$, seen as the rows of a matrix $\evMat \in
\matSpace[\rdim][\order]$. We also have a \emph{multiplication matrix} $\mulmat
\in \matSpace[\order]$ which specifies the multiplication of vectors $\row{e}
\in \matSpace[1][\order]$ by polynomials $p \in \polRing$ as $p \mul \row{e} =
\row{e}\, p(\mulmat)$. Then, we want to find $\polRing$-linear relations
between these vectors, that is, some $\int = (p_1,\ldots,p_\rdim) \in
\polRing^\rdim$ such that $\int \mul \evMat = p_1 \mul \row{e}_1 + \cdots +
p_\rdim \mul \row{e}_\rdim = 0$. Such a $\int$ is called an \emph{interpolant
for $(\evMat,\mulmat)$}.

Hereafter, the matrix $\mulmat$ is in Jordan canonical form: this assumption is
satisfied in many interesting applications, as explained below. The notion of
interpolant we consider is directly related to the one introduced
in~\cite{BarBul92,BecLab00}. Suppose that $\mulmat$ has $\nbbl$ Jordan blocks
of dimensions $\szbl_1\times\szbl_1, \ldots, \szbl_\nbbl\times\szbl_\nbbl$ and
with respective eigenvalues $\evpt_1, \ldots, \evpt_\nbbl$; in particular,
$\order = \szbl_1 + \cdots + \szbl_\nbbl$.  Then, one may identify
$\field^\order$ with \vspace{-0.05cm}\[\polSpace=\polRing/(X^{\szbl_1}) \times
\cdots \times \polRing/(X^{\szbl_\nbbl}),\vspace{-0.05cm}\] by mapping any
$\row{f}=(f_1,\dots,f_\nbbl)$ in $\polSpace$ to the vector $\row{e} \in
\field^\order$ made from the concatenation of the coefficient vectors of
$f_1,\dots,f_\nbbl$. Over $\polSpace$, the $\polRing$-module structure on
$\field^\order$ given by $p \mul \row{e} = \row{e}\, p(\mulmat)$ becomes
\vspace{-0.05cm}\[p \cdot \row{f} = (p(X+\evpt_1) f_1 \bmod X^{\szbl_1},
\dots,p(X+\evpt_\nbbl) f_\nbbl \bmod X^{\szbl_n}).\vspace{-0.05cm}\] Now, if
$(\row{e}_1,\dots,\row{e}_\rdim) \in \field^{\rdim \times \sigma}$ is
associated to $(\row{f}_1,\dots,\row{f}_\rdim) \in \polSpace^\rdim$, with
$\row{f}_i = ({f}_{i,1},\dots,{f}_{i,\nbbl})$ and ${f}_{i,j}$ in
$\polRing/(X^{\szbl_j})$ for all $i,j$, the relation $p_1 \mul \row{e}_1 +
\cdots + p_\rdim \mul \row{e}_\rdim = 0$ means that for all $j$ in
$\{1,\dots,\nbbl\}$, we have \vspace{-0.05cm}\[p_1(X+\evpt_j) f_{1,j} + \cdots
  + p_\rdim(X+\evpt_j) f_{\rdim,j} = 0 \bmod X^{\sigma_j};\vspace{-0.05cm}\]
  applying a translation by $-\evpt_j$, this is equivalent to
\vspace{-0.05cm}
\begin{equation*}
p_1 f_{1,j}(X-x_j) + \cdots + p_\rdim
f_{\rdim,j}(X-x_j) = 0 \bmod (X-x_j)^{\sigma_j}.  
\end{equation*}
\vspace{-0.05cm}
Thus, in terms of vector M-Pad\'e approximation as in~\cite{BarBul92,BecLab00},
$(p_1,\ldots,p_\rdim)$ is an interpolant for
$(\row{f}_1,\ldots,\row{f}_\rdim)$, $\evpt_1, \ldots, \evpt_\nbbl$, and
$\szbl_1, \ldots, \szbl_\nbbl$.

The set of all interpolants for $(\evMat,\mulmat)$ is a free $\polRing$-module
of rank $\rdim$. We are interested in computing a basis of this module,
represented as a matrix in $\intSpace$ and called an \emph{interpolation basis
for $(\evMat,\mulmat)$}. Its rows are interpolants for $(\evMat,\mulmat)$, and
any interpolant for $(\evMat,\mulmat)$ can be written as a unique
$\polRing$-linear combination of its rows.

Besides, we look for interpolants that have some type of minimal degree.
Following~\cite{BarBul92,ZhoLab12}, for a nonzero $\int = [p_1,\ldots,p_\rdim]
\in \polMatSpace[1][\rdim]$ and a \emph{shift} $\shifts =
(\shift{1},\ldots,\shift{\rdim}) \in \shiftSpace$, we define the
\emph{$\shifts$-degree} of $\int$ as $\max_{1\le j\le \rdim} (\deg(p_j) +
\shift{j})$. Up to a change of sign, this notion of $\shifts$-degree is
equivalent to the one in~\cite{BeLaVi06} and to the notion of \emph{defect}
from~\cite[Definition~3.1]{BecLab94}.

Then, the \emph{$\shifts$-row degree} of a matrix $\mat{P} \in
\polMatSpace[k][\rdim]$ of rank~$k$ is the tuple $\rdeg[\shifts]{\mat{P}} =
(d_1,\ldots,d_k) \in \shiftSpace[k]$ with $d_i$ the $\shifts$-degree of the
$i$-th row of $\mat{P}$. The \emph{$\shifts$-leading matrix} of $\mat{P} =
[p_{ij}]_{i,j}$ is the matrix in $\matSpace[k][\rdim]$ whose entry $(i,j)$ is
the coefficient of degree $d_i - \shift{j}$ of $p_{ij}$.  Then, $\mat{P}$ is
\emph{$\shifts$-reduced} if its $\shifts$-leading matrix has rank~$k$;
see~\cite{BeLaVi06}.

Our aim is to compute an \emph{$\shifts$-minimal} interpolation basis for
$(\evMat,\mulmat)$, that is, one which is $\shifts$-reduced: equivalently, it
is an interpolation basis whose $\shifts$-row degree, once written in
nondecreasing order, is lexicographically minimal. This corresponds to
Problem~\ref{pbm:mib} below. In particular, an interpolant of minimal degree
can be read off from an $\shifts$-minimal interpolation basis for the
\emph{uniform} shift $\shifts = \unishift$.

\begin{center}
\fbox{ \begin{minipage}{7.9cm}
\begin{pbm}[Minimal interpolation basis]
  \label{pbm:mib}
~\\
\emph{Input:}
\begin{itemize}
  \item the base field $\field$,
  \item the dimensions $\rdim$ and $\order$,
  \item a matrix $\evMat \in \evSpace{\order}$,
  \item a Jordan matrix $\mulmat \in \matSpace[\order]$,
  \item a shift $\shifts\in\shiftSpace$.
\end{itemize}

\emph{Output:} an $\shifts$-minimal interpolation basis
for $(\evMat,\mulmat)$.
\end{pbm}
\end{minipage}
}
\end{center}

A well-known particular case of this problem is Hermite-Pad\'e approximation,
that is, the computation of \emph{order bases} (or $\order$-bases, or minimal
approximant bases), where $\mulmat$ has only eigenvalue $0$. Previous work on
this case includes~\cite{BecLab94,GiJeVi03,Storjohann06,ZhoLab12} with
algorithms focusing on $\mulmat$ with $\cdim$ blocks of identical size
$\order/\cdim$. For a shift $\shifts \in \NN^\rdim$ with nonnegative entries,
we write $\sshifts$ for the sum of its entries. Then, in this context, the cost
bound $\softO{ \rdim^{\expmatmul-1} \order}$ has been obtained under each of
the following assumptions:
\begin{enumerate}[($H_1$)]
  \item $\max(\shifts) - \min(\shifts) \in \bigO{\order/\rdim}$ in
\cite[Theorem~5.3]{ZhoLab12} and more generally
$\sshifts[\shifts-\min(\shifts)] \in \bigO{\order}$ in
\cite[Section~4.1]{Zhou12};
  \item $\sshifts[\!\max(\shifts)-\shifts] \in \bigO{\order}$ in
\cite[Theorem~6.14]{ZhoLab12}.
\end{enumerate}
These assumptions imply in particular that any $\shifts$-minimal basis has size
in $\bigO{\rdim \order}$, where by \emph{size} we mean the number of field
elements used to represent the matrix.

An interesting example of a shift not covered by ($H_1$) or ($H_2$) is
$\shifts[h] = (0,\order,2\order,\ldots,(\rdim-1)\order)$ which is related to
the Hermite form \cite[Lemma~2.6]{BeLaVi06}. In general, as detailed in
Appendix~\ref{app:shift}, one may assume without loss of generality that
$\min(\shifts) = 0$, $\max(\shifts) \in \bigO{\rdim \order}$, and $\sshifts \in
\bigO{\rdim^2 \order}$. 

There are also applications of Problem~\ref{pbm:mib} to multivariate
interpolation, where $\mulmat$ is not nilpotent anymore, and for which we have
neither $(H_1)$ nor $(H_2)$, as we will see in
Subsection~\ref{subsec:applications}. It was left
as an open problem in~\cite[Section~7]{ZhoLab12} to obtain algorithms with cost
bound $\softO{\rdim^{\expmatmul-1} \order}$ for such matrices $\mulmat$ and for
arbitrary shifts. In this paper, we solve this open problem.

An immediate challenge is that for an arbitrary shift $\shifts$, the size of an
$\shifts$-minimal interpolation basis may be beyond our target cost: we show
this in Appendix~\ref{app:large_size_mib} with an example of Hermite-Pad\'e
approximation. Our answer is to compute a basis in \emph{$\shifts$-Popov form}:
among its many interesting features, it can be represented using at most $\rdim
(\order+1)$ elements from $\field$, and it is canonical: for every nonsingular
$\mat{A} \in \polMatSpace[\rdim]$ and $\shifts \in \shiftSpace$, there is a
unique matrix $\mat{P}$ in $\shifts$-Popov form which is left-unimodularly
equivalent to $\mat{A}$. We use the definition from~\cite[Section~7]{BecLab00},
phrased using the notion of pivot~\cite[Section~6.7.2]{Kailath80}.

\begin{dfn}[Pivot of a row]
\label{dfn:pivot}
Let $\row{p} = [p_j]_j \in \polMatSpace[1][\rdim]$ be a nonzero row vector and
let $\shifts \in \shiftSpace$. The \emph{$\shifts$-pivot index} of $\row{p}$ is
the largest index $j \in \{1,\ldots,\rdim\}$ such that $\rdeg[\shifts]{\row{p}}
= \deg(p_j) + \shift{j}$; then, $p_j$ and $\deg(p_j)$ are called the
\emph{$\shifts$-pivot entry} and the \emph{$\shifts$-pivot degree} of
$\row{p}$.
\end{dfn}

\begin{dfn}[Popov form]
\label{dfn:popov}
Let $\mat{P} \in \polMatSpace[\rdim]$ be nonsingular and let $\shifts \in
\shiftSpace$. Then, $\mat{P}$ is said to be in \emph{$\shifts$-Popov form} if
its $\shifts$-pivot entries are monic and on its diagonal, and if in each
column of $\mat{P}$ the nonpivot entries have degree less than the pivot entry.
\end{dfn}

We call \emph{$\shifts$-Popov interpolation basis for $(\evMat,\mulmat)$} the
unique interpolation basis for $(\evMat,\mulmat)$ which is in $\shifts$-Popov
form; in particular, it is an $\shifts$-minimal one. For small values of
$\order$, namely $\order \in \bigO{\rdim}$, we gave
in~\cite[Section~7]{JeNeScVi15} an algorithm which computes the $\shifts$-Popov
interpolation basis in $\softO{\order^{\expmatmul-1} \rdim}$ operations for an
arbitrary~$\shifts$ \cite[Theorem~1.4]{JeNeScVi15}. Hence, in what follows, we
focus on the case $\rdim \in \bigO{\order}$.

We use the convenient assumption that $\mulmat$ is given to us as a
list of eigenvalues and block sizes:
\begin{equation*}
\mulmat=((x_1,\sigma_{1,1}),\dots,(x_1,\sigma_{1,r_1}),\dots,(x_t,\sigma_{t,1}),\dots,(x_t,\sigma_{t,r_t})),  
\end{equation*}
for some pairwise distinct eigenvalues $x_1,\dots,x_t$, with $r_1 \ge \cdots
\ge r_t$ and $\sigma_{i,1} \ge \cdots \ge \sigma_{i,r_i}$ for all $i$; we say
that this representation is \emph{standard}.

\begin{thm}
  \label{thm:pib}
  Assuming that $\mulmat\in\matSpace[\order]$ is a Jordan matrix given by a
  standard representation, there is a deterministic algorithm which solves
  Problem~\ref{pbm:mib} using
  \begin{align*}
  \bigO{ & \rdim^{\expmatmul-1} \polmultime{\order}
  \log(\order) \log(\order/\rdim)^2 } & \text{if }
  \expmatmul>2, \\
  \bigO{ & \rdim \polmultime{\order} \log(\order)
  \log(\order/\rdim)^2 \log(\rdim)^3 
  } & \text{if } \expmatmul=2\;
  \end{align*}
  operations in $\field$ and returns the $\shifts$-Popov interpolation basis
  for $(\evMat,\mulmat)$.
\end{thm}

In this result, $\polmultime{\cdot}$ is such that polynomials of degree at most
$d$ in $\polRing$ can be multiplied using $\polmultime{d}$ operations in
$\field$, and $\polmultime{\cdot}$ satisfies the super-linearity properties
of~\cite[Chapter~8]{vzGathen13}. It follows from~\cite{CanKal91} that
$\polmultime{d}$ can be taken in $\bigO{d \log(d) \log(\log(d))}$. The exponent
$\expmatmul$ is so that we can multiply $\rdim \times \rdim$ matrices in
$\bigO{\rdim^\expmatmul}$ ring operations on any ring, the best known bound
being $\expmatmul < 2.38$~\cite{CopWin90, LeGall14}.

Compared to our work in~\cite{JeNeScVi15}, our algorithm here has two key new
features:
\begin{itemize}
  \setlength\itemsep{0cm}
  \item it supports arbitrary shifts with a cost $\softO{\rdim^{\expmatmul-1}
  \order}$;
  \item it computes the basis in $\shifts$-Popov form.
\end{itemize}
To the best of our knowledge, no algorithm for Problem~\ref{pbm:mib} with cost
$\softO{\rdim^{\expmatmul-1} \order}$ was known previously for \emph{arbitrary}
shifts, even for the specific case of order basis computation.

If $\mulmat$ is given as an arbitrary list $( (\evpt_1,\szbl_1), \ldots,
(\evpt_\nbbl,\szbl_\nbbl) )$, we can reorder it (and permute the columns of
$\evMat$ accordingly) to obtain an equivalent standard representation in time
$\bigO{\polmultime{\order} \log(\order)^3}$~\cite[Proposition~12]{BoJeSc08}; if
$\field$ is equipped with an order, and if we assume that comparisons take unit
time, this can of course be done in time $\bigO{\order \log(\order)}$.

\subsection{Overview of our approach}
\label{subsec:overview} 

Several previous algorithms for order basis computation, such as those
in~\cite{BecLab94,GiJeVi03}, follow a divide-and-conquer scheme inspired by the
Knuth-Sch\"onhage-Moenck algorithm~\cite{Knu70,Sch71,Moe73}
This paper builds on our previous work in~\cite{JeNeScVi15}, where we extended
this recursive approach to more general interpolation problems.  However, the
main algorithm in~\cite{JeNeScVi15} does not handle an arbitrary shift
$\shifts$ with a satisfactory complexity; here, we use it as a black box, after
showing how to reduce the  problem to a new one with suitable shift. 

Let $\evMat$, $\mulmat$, and $\shifts$ be our input, and write $\mulmat^{(1)}$
and $\mulmat^{(2)}$ for the $\order/2\times\order/2$ leading and trailing
principal submatrices of~$\mulmat$. First, compute an $\shifts$-minimal
interpolation basis $\intBasis^{(1)}$ for $\mulmat^{(1)}$ and the first
$\order/2$ columns of $\evMat$; then, compute the last $\order/2$ columns
$\evMat^{(2)}$ of the \emph{residual} $\intBasis^{(1)} \mul \evMat$; then,
compute a $\shifts[t]$-minimal interpolation basis $\intBasis^{(2)}$ for
$(\evMat^{(2)},\mulmat^{(2)})$ with $\shifts[t] =
\rdeg[\shifts]{\intBasis^{(1)}}$; finally, return the matrix product
$\intBasis^{(2)} \intBasis^{(1)}$. 

This approach allows  to solve Problem~\ref{pbm:mib} using
$\softO{\rdim^\expmatmul \order}$ operations in $\field$. In the case of
Hermite-Pad\'e approximation, this is the divide-and-conquer algorithm
in~\cite{BecLab94}. Besides, an $\shifts$-minimal basis computed by this
method has degree at most $\order$ and thus size in $\bigO{\rdim^2 \order}$,
and there are indeed instances of Problem~\ref{pbm:mib} for which this size
reaches $\Theta(\rdim^2 \order)$. In Appendix~\ref{app:large_size_mib}, we show
such an instance for the algorithm in~\cite{BecLab94}, in the case of
Hermite-Pad\'e approximation.

It is known that the average degree of the rows of any $\shifts$-minimal
interpolation basis is at most $(\order + \xi)/\rdim$, where $\xi =
\sshifts[\shifts-\min(\shifts)]$~\cite[Theorem~4.1]{BarBul92}.
In~\cite{JeNeScVi15}, focusing on the case where $\xi$ is small compared to
$\order$, and preserving such a property in recursive calls via changes of
shifts, we obtained the cost bound 
\begin{equation}\label{eq:costJeNeScVi15}
   \bigO{  \rdim^{\expmatmul-1} \polmultime{\order} \log(\order)
   \log(\order/\rdim) + \rdim^{\expmatmul-1} \polmultime{\xi} \log(\xi/\rdim)}
\end{equation}
to solve Problem~\ref{pbm:mib}; this cost is for $\expmatmul>2$, and a similar
one holds for $\expmatmul=2$, both being in $\softO{\rdim^{\expmatmul-1}
(\order + \xi)}$. The fundamental reason for this kind of improvement over
$\softO{\rdim^\expmatmul \order}$, already seen with~\cite{ZhoLab12},  is that
one controls the average row degree of the bases $\intBasis^{(2)}$ and
$\intBasis^{(1)}$, \emph{and} of their product $\intBasis^{(2)}
\intBasis^{(1)}$.

This result is $\softO{\rdim^{\expmatmul-1} \order}$ for $\xi$ in $\bigO
\order$. The main difficulty to extend it to any shift $\shifts$ is to control
the size of the computed bases: the Hermite-Pad\'e example pointed out above
corresponds to $\xi =\Theta(\rdim\order)$ and leads to an output of size
$\Theta(\rdim^2 \order)$ for the algorithm of~\cite{JeNeScVi15} as well.

The key ingredient to control this size is to work with bases in
$\shifts$-Popov form: for any $\shifts$, the $\shifts$-Popov interpolation
basis $\intBasis$ for $(\evMat,\mulmat)$ has average \emph{column} degree at
most $\order/\rdim$ and size at most $\rdim(\order+1)$, as detailed in
Section~\ref{sec:algo}.

Now, suppose that we have computed recursively the bases $\intBasis^{(2)}$ and
$\intBasis^{(1)}$ in $\shifts$- and $\shifts[t]$-Popov form; we want to output
the $\shifts$-Popov form $\intBasis$ of $\intBasis^{(2)} \intBasis^{(1)}$. In
general, this product is not normalized and may have size $\Theta(\rdim^2
\order)$: its computation is beyond our target cost. Thus, one main idea is
that we will \emph{not} rely on polynomial matrix multiplication to combine the
bases obtained recursively; instead, we use a minimal interpolation basis
computation for a shift that has good properties as explained below.

An important remark is that if we know \emph{a priori} the column degree
$\minDegs$ of $\intBasis$, then the problem becomes easier. This idea was
already used in algorithms for the Hermite form~$\mat{H}$ of a polynomial
matrix~\cite{GupSto11,Zhou12}, which first compute the column degree $\minDegs$
of $\mat{H}$, and then obtain $\mat{H}$ as a submatrix of some minimal
nullspace basis for a shift involving $-\minDegs$.

In Section~\ref{sec:known_degrees}, we study the problem of computing the
$\shifts$-Popov interpolation basis $\intBasis$ for $(\evMat,\mulmat)$ having
its column degree $\minDegs$ as an additional input. We show that this reduces
to the computation of a $\shifts[d]$-minimal interpolation basis $\mat{R}$ with
the specific shift $\shifts[d] = -\minDegs$. The properties of this shift~$\shifts[d]$ 
allow us first to compute $\mat{R}$ in
$\softO{\rdim^{\expmatmul-1}\order}$ operations using the partial linearization
framework from~\cite[Section~3]{Storjohann06} and the minimal interpolation
basis algorithm in~\cite[Section~3]{JeNeScVi15}, and second to easily retrieve
$\intBasis$ from $\mat{R}$.

Still, in general we do not know $\minDegs$. We will thus compute it, relying on a variation of 
the divide-and-conquer strategy at the beginning of this subsection.  We
stop the recursion as soon as $\order \le \rdim$, in which case we do not need~$\minDegs$ 
to achieve efficiency: the algorithm
from~\cite[Section~7]{JeNeScVi15} computes the $\shifts$-Popov interpolation
basis in $\softO{\order^{\expmatmul-1} \rdim}$ operations 
for any~$\shifts$~\cite[Theorem~1.4]{JeNeScVi15}. Then, we show in
Section~\ref{sec:mindeg_invariant} that from $\intBasis^{(1)}$ and
$\intBasis^{(2)}$ computed recursively \emph{in shifted Popov form}, we can
obtain $\minDegs$ for free. Finally, instead of considering $\intBasis^{(2)}
\intBasis^{(1)}$, we use the knowledge of $\minDegs$ to compute the basis
$\intBasis$ from scratch as explained in the previous paragraph.

This summarizes our main algorithm, which is presented in
Section~\ref{sec:algo}.

\subsection{Previous work and applications}
\label{subsec:applications}

As a particular case of Problem~\ref{pbm:mib}, when all the eigenvalues of
$\mulmat$ are zero, we obtain the following complexity result about \emph{order
basis} computation~\cite[Definition~2.2]{ZhoLab12}.

\begin{thm}
  \label{thm:order_basis}
  Let $\rdim,\cdim \in \ZZp$, let $(\order_1,\ldots,\order_\cdim) \in
  \ZZp^\cdim$, let $\shifts \in \shiftSpace$, and let $\mat{F} \in
  \polMatSpace[\rdim][\cdim]$
  with its $j$-th column $\matcol{\mat{F}}{j}$ of
  degree less than $\order_j$. The unique basis $\intBasis \in \intSpace$ in
  $\shifts$-Popov form of the $\polRing$-module of approximants
  \[ \{ \row{p} \in \polMatSpace[1][\rdim] \;\mid\;
    \row{p} \matcol{\mat{F}}{j} = 0 \bmod X^{\order_j} \text{ for each } j \}
  \] can be computed deterministically
  using
  \begin{align*}
  \bigO{ & \rdim^{\expmatmul-1} \polmultime{\order}
  \log(\order) \log(\order/\rdim)^2 } & \text{if }
  \expmatmul>2, \\
  \bigO{ & \rdim \polmultime{\order} \log(\order)
  \log(\order/\rdim)^2 \log(\rdim)^3 
  } & \text{if } \expmatmul=2\;
  \end{align*}
  operations in $\field$, where $\order = \order_1 + \cdots + \order_\cdim$.
\end{thm}

Previous work on this problem
includes~\cite{BecLab94,GiJeVi03,Storjohann06,ZhoLab12,JeNeScVi15}, mostly with
identical orders $\order_1 = \cdots = \order_\cdim$; an interesting particular
case is Hermite-Pad\'e approximation with $\cdim=1$. To simplify matters, for
all our comparisons, we consider $\expmatmul>2$. For order basis computation
with $\order_1 = \cdots = \order_\cdim$ and $\cdim \le \rdim$, the cost bound
$\bigO{ \rdim^{\expmatmul} \polmultime{\order/\rdim} \log(\order/\cdim) }$ was
achieved in~\cite{ZhoLab12} under either of the assumptions $(H_1)$ and $(H_2)$
on the shift. Still, the corresponding algorithm returns a basis $\intBasis$
which is only $\shifts$-reduced, and because both the shift $\shifts$ and the
degrees in $\intBasis$ may be unbalanced, one cannot directly rely on the
fastest known normalization algorithm~\cite{SarSto11} to compute the
$\shifts$-Popov form of $\intBasis$ within the target cost.

Another application of Problem~\ref{pbm:mib} is a multivariate interpolation
problem that arises for example in the first step of algorithms for the
list-decoding of Parvaresh-Vardy codes~\cite{ParVar05} and of folded
Reed-Solomon codes~\cite{GurRud08}, as well as in robust Private Information
Retrieval~\cite{DeGoHe12a}. The bivariate case corresponds to the interpolation
steps of K\"otter and Vardy's soft-decoding~\cite{KoeVar03a} and Guruswami and
Sudan's list-decoding~\cite{GurSud99} algorithms for Reed-Solomon codes.

Given a set of points in $\field^{\nvars+1}$ and associated multiplicities,
this problem asks to find a multivariate polynomial $Q(X,Y_1,\ldots,Y_\nvars)$
such that: $(a)$ $Q$ has prescribed exponents for the $Y$ variables, so that
the problem can be linearized with respect to $Y$, leaving us with a linear
algebra problem over $\polRing$; $(b)$ $Q$ vanishes at all the given points
with their multiplicities, inducing a structure of $\polRing$-module on the set
of solutions; $(c)$ $Q$ has some type of minimal weighted degree, which can be
seen as the minimality of the shifted degree of the vector over $\polRing$ that
represents $Q$.

Following the coding theory context~\cite{GurSud99,ParVar05}, given a point
$(x,y) \in \field \times \field^{\nvars}$ and a set of exponents $\supp \subset
\NN^{\nvars+1}$, we say that the polynomial $Q(X,Y) \in
\field[X,Y_1,\ldots,Y_\nvars]$ \emph{vanishes at $(x,y)$ with multiplicity
support $\supp$} if the shifted polynomial $Q(X+x, Y+y)$ has no monomial with
exponent in $\supp$. We will only consider supports that are \emph{stable under
division}, meaning that if $(\gamma_0,\gamma_1,\ldots,\gamma_\nvars)$ is in
$\supp$, then any $(\gamma'_0,\gamma'_1,\ldots,\gamma'_\nvars)$ with $\gamma'_j
\le \gamma_j$ for all $j$ is also in $\supp$. 

Now, given a set of exponents $\expSet \subset \NN^\nvars$, we represent
$Q(X,Y) = \sum_{\gamma \in \expSet} p_\gamma Y^\gamma$ as the row $\row{p} =
[p_\gamma]_{\gamma \in \expSet} \in \polMatSpace[1][\rdim]$ where $\rdim$ is
the cardinality of $\expSet$. Again, we assume that the exponent set $\expSet$
is stable under division; then, the set of solutions is a free
$\polRing$-module of rank $\rdim$. In the mentioned applications, we typically
have $\expSet = \{ (\gamma_1,\ldots,\gamma_\nvars) \in \NN^\nvars \;\mid\;
\gamma_1 + \cdots + \gamma_\nvars \le \ell\}$ for an integer $\ell$ called the
\emph{list-size parameter}.

Besides, we are given some weights $\shifts[w] =
(\shift[w]{1},\ldots,\shift[w]{\nvars}) \in \NN^\nvars$ on the variables $Y =
Y_1,\ldots,Y_\nvars$, and we are looking for $Q(X,Y)$ which has minimal
$\shifts[w]$-weighted degree, which is the degree in $X$ of the polynomial
\begin{align*}
 Q(X,& X^{\shift[w]{1}}Y_1,\ldots,X^{\shift[w]{\nvars}}Y_\nvars) \\
 & = \sum_{\gamma \in \expSet} p_\gamma X^{\gamma_1 \shift[w]{1} + \cdots +
 \gamma_\nvars \shift[w]{\nvars}} Y_1^{\gamma_1} \cdots
 Y_\nvars^{\gamma_\nvars}.
\end{align*}
This is exactly requiring that the $\shifts$-degree of $\row{p} =
[p_\gamma]_\gamma$ be minimal, for $\shifts = [\gamma_1 \shift[w]{1} + \cdots +
\gamma_\nvars \shift[w]{\nvars}]_\gamma$. We note that it is sometimes
important, for example in~\cite{DeGoHe12a}, to return a whole $\shifts$-minimal
interpolation basis and not only one interpolant of small $\shifts$-degree. 

\begin{figure}[ht]
  \centering
\fbox{ \begin{minipage}{7.9cm}
\begin{pbm}[Multivariate interpolation]
  \label{pbm:multi_int}
~\\
\emph{Input:}
  \begin{itemize}
  \item number of $Y$ variables $\nvars>0$,
  \item set $\expSet \subset \NN^\nvars$ of cardinality $\rdim$,
  stable under division,
  \item pairwise distinct points  $\{ (\evpt_k,y_k) \in \field \times \field^{\nvars} \}_{1 \le
  k \le \nbpt}$,
  \item supports $\{ \supp_k \subset \NN^{\nvars+1} \}_{1 \le k \le \nbpt}$, stable under division,
  \item a shift $\shifts\in\shiftSpace$.
  \end{itemize}

\emph{Output:} a matrix $\intBasis \in \intSpace$ such that 
\begin{itemize}
  \item the rows of $\intBasis$ form a basis of the $\polRing$-module
    \begin{align*}
      \Bigg\{ \row{p} & =
      [p_\gamma]_{\gamma\in\expSet} \in \polMatSpace[1][\rdim]
      \;\bigg|\; \sum_{\gamma \in \expSet}
      p_\gamma(X) Y^\gamma \text{ vanishes} \\
      & \text{at } (\evpt_k,y_k) \text{ with
      support } \supp_k \text{ for } 1 \le k \le \nbpt \Bigg\} ,
    \end{align*}
  \item $\intBasis$ is $\shifts$-reduced.
\end{itemize}
\end{pbm}
\end{minipage}
}
\end{figure}

For more details about the reduction from Problem~\ref{pbm:multi_int} to
Problem~\ref{pbm:mib}, explaining how to build the input matrices
$(\evMat,\mulmat)$ with $\mulmat$ a Jordan matrix in standard representation,
we refer the reader to~\cite[Subsection~2.4]{JeNeScVi15}. In particular, the
dimension $\order$ is the sum of the cardinalities of the multiplicity
supports. In the mentioned applications to coding theory, we have $\rdim =
\binom{\nvars+\ell}{\nvars}$ where $\ell$ is the list-size parameter; and
$\order$ is the so-called \emph{cost} in the soft-decoding
context~\cite[Section~III]{KoeVar03a}, that is, the number of linear equations
when linearizing the problem over $\field$. As a consequence of
Theorem~\ref{thm:pib}, we obtain the following complexity result.

\begin{thm}
  \label{thm:multi_int} 
  Let $\order = \sum_{1\le k\le p} \#\supp_k$. There is a deterministic
  algorithm which solves Problem~\ref{pbm:multi_int} using
  \begin{align*}
  \bigO{ & \rdim^{\expmatmul-1} \polmultime{\order}
  \log(\order) \log(\order/\rdim)^2 } & \text{if }
  \expmatmul>2, \\
  \bigO{ & \rdim \polmultime{\order} \log(\order)
  \log(\order/\rdim)^2 \log(\rdim)^3 
  } & \text{if } \expmatmul=2\;
  \end{align*}
  operations in $\field$, and returns the unique basis of solutions which is in
  $\shifts$-Popov form.
\end{thm}

Under the assumption that the $x_k$ are pairwise distinct, the cost bound
$\bigO{ \rdim^{\expmatmul-1} \polmultime{\order} \log(\order)^2 }$ was achieved
for an arbitrary shift using fast structured linear
algebra~\cite[Theorems~1~and~2]{CJNSV15}, following work
by~\cite{OlsSho99,RotRuc00,ZeGeAu11}. However, the corresponding algorithm is
randomized and returns only one interpolant of small $\shifts$-degree. For a
broader overview of previous work on this problem, we refer the reader to the
introductive sections of~\cite{BeeBra10,CJNSV15} and to
\cite[Section~2]{JeNeScVi15}.

The term $\bigO{\rdim^{\expmatmul-1} \polmultime{\xi} \log(\xi/\rdim)}$
reported in~\eqref{eq:costJeNeScVi15} for the cost of the algorithm
of~\cite{JeNeScVi15} can be neglected if $\xi \in \bigO{\order}$; this is for
instance satisfied in the context of bivariate interpolation for soft- or
list-decoding of Reed-Solomon codes \cite[Sections~2.5 and~2.6]{JeNeScVi15}.
However, we do not have this bound on $\xi$ in the list-decoding of
Parvaresh-Vardy codes and folded Reed-Solomon codes and in Private Information
Retrieval. Thus, in these cases our algorithm achieves the best known cost
bound, improving upon~\cite{PB08,KB10,CohHen12,DeGoHe12a,JeNeScVi15}.

\section{Fast Popov interpolation basis}
\label{sec:algo}
In this section, we present our main result, Algorithm~\ref{algo:pib}. It
relies on three subroutines; two of them are from~\cite{JeNeScVi15}, while the
third is a key new ingredient, detailed in Section~\ref{sec:known_degrees}.
\begin{itemize}
  \item \algoname{LinearizationMIB}~\cite[Algorithm~9]{JeNeScVi15}
    solves the base case $\order \le \rdim$ using linear algebra over $\field$.
    The inputs are $\evMat$, $\mulmat$, $\shifts$, as well as an integer
    for which we can take the first power of two greater than or equal
    to $\order$.
  \item \algoname{ComputeResiduals}~\cite[Algorithm~5]{JeNeScVi15}
    (with an additional pre-processing detailed at the end of
    Section~\ref{sec:known_degrees}) computes the residual $\intBasis^{(1)}
    \mul \evMat$ from the first basis $\intBasis^{(1)}$ obtained recursively.
  \item \algoname{KnownMinDegMIB}, detailed in Section~\ref{sec:known_degrees},
    computes the $\shifts$-Popov interpolation basis when one knows \emph{a
    priori} the $\shifts$-minimal degree of $(\evMat,\mulmat)$ (see below).
\end{itemize}

In what follows, by \emph{$\shifts$-minimal degree of $(\evMat,\mulmat)$} we
mean the tuple of degrees of the diagonal entries of the $\shifts$-Popov
interpolation basis $\intBasis$ for $(\evMat,\mulmat)$. Because $\intBasis$ is
in $\shifts$-Popov form, this is also the column degree of $\intBasis$, and the
sum of these degrees is $\deg(\det(\intBasis))$. As a consequence, using
Theorem~4.1 in~\cite{BarBul92} (or following the lines of \cite{Kailath80} and
\cite{BecLab00}) we obtain the following lemma, which implies in particular
that the size of $\intBasis$ is at most $\rdim (\order+1)$.

\begin{lem}
  \label{lem:mindeg}
  Let $\evMat \in \evSpace{\order}$, $\mulmat \in \matSpace[\order]$, $\shifts
  \in \shiftSpace$, and let $(\minDeg_1,\ldots,\minDeg_\rdim)$ be the
  $\shifts$-minimal degree of $(\evMat,\mulmat)$. Then, we have $\minDeg_1 +
  \cdots + \minDeg_\rdim \le \order$.
\end{lem}

\begin{figure}[!h]
\centering
\fbox{\begin{minipage}{\linewidth}
\begin{algo} \algoname{PopovMIB}
  \label{algo:pib}
~\\
\emph{Input:} 
  \begin{itemize}
  \item a matrix $\evMat\in\evSpace{\order}$,
  \item a Jordan matrix $\jordan \in \matSpace[\order]$ in standard representation,
  \item a shift $\shifts\in\shiftSpace$.
  \end{itemize}
\emph{Output:}
\begin{itemize}
  \item the $\shifts$-Popov interpolation basis $\intBasis$ for
    $(\evMat,\mulmat)$,
  \item the $\shifts$-minimal degree $\minDegs = (\minDeg_1,\ldots,\minDeg_\rdim)$ of
    $(\evMat,\mulmat)$.
\end{itemize}
\begin{enumerate}[{\bf 1.}] 
  \item If $\order\le\rdim$, return $\algoname{LinearizationMIB}(\evMat,\mulmat,\shifts,2^{\lceil\log_2(\order)\rceil})$
\item Else
  \begin{enumerate}[{\bf a.}]
    \item $\evMat^{(1)} \leftarrow$ first $\lceil\order/2\rceil$ columns of $\evMat$
    \item $(\mat{P}^{(1)}, \minDegs^{(1)}) \leftarrow \algoname{PopovMIB}(\evMat^{(1)},\mulmat^{(1)},\shifts)$
    \item $\evMat^{(2)} \leftarrow$ last $\lfloor\order/2\rfloor$ columns of \\
      ${ }\:\:\quad\qquad \mat{P}^{(1)} \mul \evMat =
      \algoname{ComputeResiduals}(\mulmat,\mat{P}^{(1)},\evMat)$
    \item $(\mat{P}^{(2)}, \minDegs^{(2)}) \leftarrow \algoname{PopovMIB}( \evMat^{(2)}, \mulmat^{(2)}, \shifts+\minDegs^{(1)} )$
    \item $\mat{P} \leftarrow \algoname{KnownMinDegMIB}(\evMat,\mulmat,\shifts,\minDegs^{(1)} + \minDegs^{(2)})$
    \item Return $(\mat{P},\minDegs^{(1)} + \minDegs^{(2)})$
  \end{enumerate}
\end{enumerate}
\end{algo}
\end{minipage}}
\end{figure}

Taking for granted the results in the next sections, we now prove our main
theorem.

\begin{proof}[of Theorem~\ref{thm:pib}]
For the case $\order \le \rdim$, the correctness and the cost bound of
Algorithm~\ref{algo:pib} both follow from~\cite[Theorem~1.4]{JeNeScVi15}: it
uses $\bigO{\order^{\expmatmul-1}\rdim + \order^\expmatmul \log(\order)}$
operations (with an extra $\log(\order)$ factor if $\expmatmul=2$).

Now, we consider the case $\order > \rdim$. Using the notation in the
algorithm, assume that $\intBasis^{(1)}$ is the $\shifts$-Popov interpolation
basis for $(\evMat^{(1)},\mulmat^{(1)})$, and $\intBasis^{(2)}$ is the
$\shifts[t]$-Popov interpolation basis for $(\evMat^{(2)},\mulmat^{(2)})$,
where $\shifts[t] = \shifts + \minDegs^{(1)} =
\rdeg[\shifts]{\intBasis^{(1)}}$, and $\minDegs^{(1)}$ and $\minDegs^{(2)}$ are
the $\shifts$- and $\shifts+\minDegs^{(1)}$-minimal degrees of
$(\evMat^{(1)},\mulmat^{(1)})$ and $(\evMat^{(2)},\mulmat^{(2)})$,
respectively.

We claim that $\intBasis^{(2)} \intBasis^{(1)}$ is $\shifts$-reduced:
this will be proved in Lemma~\ref{lem:pivdegs}. Let us then prove that
$\intBasis^{(2)} \intBasis^{(1)}$ is an interpolation basis for
$(\evMat,\mulmat)$. Let $\row{p} \in
\polMatSpace[1][\rdim]$ be an interpolant for $(\evMat,\mulmat)$. Since
$\mulmat$ is upper triangular, $\row{p}$ is in particular an interpolant for
$(\evMat^{(1)},\mulmat^{(1)})$, so there exists $\row{v} \in
\polMatSpace[1][\rdim]$ such that $\row{p} = \row{v} \intBasis^{(1)}$.
Besides, we have $\intBasis^{(1)} \mul \evMat = [0 | \evMat^{(2)}]$, so that $0
= \row{p} \mul \evMat = \row{v} \intBasis^{(1)} \mul \evMat = [0 | \row{v} \mul
\evMat^{(2)}]$, and thus $\row{v} \mul \evMat^{(2)} = 0$.  Then, there exists
$\row{w} \in \polMatSpace[1][\rdim]$ such that $\row{v} = \row{w}
\intBasis^{(2)}$, which gives $\row{p} = \row{w} \intBasis^{(2)}
\intBasis^{(1)}$.

In particular, the $\shifts$-Popov interpolation basis for $(\evMat,\mulmat)$
is the $\shifts$-Popov form of $\intBasis^{(2)} \intBasis^{(1)}$. Thus,
Lemma~\ref{lem:pivdegs} combined with Lemma~\ref{lem:invariant} will show that the
$\shifts$-minimal degree of $(\evMat,\mulmat)$ is $\minDegs^{(1)} +
\minDegs^{(2)}$. As a result, Proposition~\ref{prop:known_mindeg_pib} states that
Step~\textbf{2.e} correctly computes the $\shifts$-Popov interpolation basis
for $(\evMat,\mulmat)$.

Concerning the cost bound, the recursion stops when $\order \le
\rdim$, and thus the algorithm uses $\bigO{\rdim^\expmatmul \log(\rdim)}$
operations (with an extra $\log(\rdim)$ factor if $\expmatmul=2$). The depth of
the recursion is $\bigO{\log(\order/\rdim)}$; we have two recursive calls in
dimensions $\rdim \times \order/2$, and two calls to subroutines with cost
bounds given in Corollary~\ref{cor:residual} and
Proposition~\ref{prop:known_mindeg_pib}, respectively. The conclusion follows
from the super-linearity properties of $\polmultime{\cdot}$.
\end{proof}

\section{Obtaining the minimal degree \\ from recursive calls}
\label{sec:mindeg_invariant}

In this section, we show that the $\shifts$-minimal degree of
$(\evMat,\mulmat)$ can be deduced for free from two bases computed recursively
as in Algorithm~\ref{algo:pib}. To do this, we actually prove a slightly more
general result about the degrees of the $\shifts$-pivot entries of so-called
weak Popov matrix forms~\cite{MulSto03}.

\begin{dfn}[Weak Popov form, pivot degree]
\label{dfn:weak_popov}
Let $\mat{P} \in \polMatSpace[\rdim]$ be nonsingular and let $\shifts \in
\shiftSpace$. Then, $\mat{P}$ is said to be in \emph{$\shifts$-weak Popov form}
if the $\shifts$-pivot indices of its rows are pairwise distinct; $\mat{P}$ is
said to be in \emph{$\shifts$-diagonal weak Popov form} if its $\shifts$-pivot
entries are on its diagonal.

If $\mat{P}$ is in $\shifts$-weak Popov form, the \emph{$\shifts$-pivot degree}
of $\mat{P}$ is the tuple $(\minDeg_1,\ldots,\minDeg_\rdim)$ where for $j \in
\{1,\ldots,\rdim\}$, $\minDeg_j$ is the $\shifts$-pivot degree of the row of
$\mat{P}$ which has $\shifts$-pivot index $j$. 
\end{dfn}

We recall from Section~\ref{sec:intro} that for $\mat{P} \in
\polMatSpace[k][\rdim]$, its $\shifts$-leading matrix
$\leadingMat[\shifts]{\mat{P}} \in \matSpace[k][\rdim]$ is formed by the
coefficients of degree $0$ of $\shiftMat{-\shifts[d]} \mat{P}
\shiftMat{\shifts}$, where $\shifts[d] = \rdeg[\shifts]{\mat{P}}$ and
$\shiftMat{\shifts}$ stands for the diagonal matrix with entries
$X^{\shift{1}}, \ldots, X^{\shift{\rdim}}$. Then, a nonsingular $\intBasis \in
\intSpace$ is in $\shifts$-diagonal weak Popov form with $\shifts$-pivot degree
$\minDegs$ if and only if $\leadingMat[\shifts]{\mat{P}}$ is lower triangular
and invertible and $\rdeg[\shifts]{\intBasis} = \shifts+\minDegs$. 

For example, at all stages of the algorithms
in~\cite{BarBul92,BecLab94,JeNeScVi15} for Problem~\ref{pbm:mib} (as well
as~\cite{GiJeVi03} if avoiding row permutations at the base case of the
recursion), the computed bases are in shifted diagonal weak Popov form. This is
due to the compatibility of this form with matrix multiplication, as stated in
the next lemma.

\begin{lem}
\label{lem:pivdegs}
Let $\shifts \in \shiftSpace$, $\intBasis^{(1)} \in \polMatSpace[\rdim]$ in
$\shifts$-diagonal weak Popov form with $\shifts$-pivot degree
$\minDegs^{(1)}$, $\shifts[t] = \shifts + \minDegs^{(1)} =
\rdeg[\shifts]{\intBasis^{(1)}}$, and $\intBasis^{(2)} \in \polMatSpace[\rdim]$
in $\shifts[t]$-diagonal weak Popov form with $\shifts[t]$-pivot degree
$\minDegs^{(2)}$. Then, $\intBasis^{(2)} \intBasis^{(1)}$ is in $\shifts$-diagonal weak Popov
form with $\shifts$-pivot degree $\minDegs^{(1)} + \minDegs^{(2)}$.
\end{lem}
\begin{proof}
By the predictable-degree property~\cite[Theorem~6.3-13]{Kailath80} we have
$\rdeg[\shifts]{\intBasis^{(2)} \intBasis^{(1)}} =
\rdeg[{\shifts[t]}]{\intBasis^{(2)}} = \shifts[t] + \minDegs^{(2)} = \shifts +
\minDegs^{(1)} + \minDegs^{(2)}$. The result follows since
$\leadingMat[\shifts]{ \intBasis^{(2)} \intBasis^{(1)} } =
\leadingMat[{\shifts[t]}]{\intBasis^{(2)}}
\leadingMat[\shifts]{\intBasis^{(1)}}$ is lower triangular and invertible.
\end{proof}

For matrices in $\shifts$-Popov form, the $\shifts$-pivot degree coincides with
the column degree: in particular, the $\shifts$-minimal degree of
$(\evMat,\mulmat)$ is the $\shifts$-pivot degree of the $\shifts$-Popov
interpolation basis for $(\evMat,\mulmat)$. With the notation of
Algorithm~\ref{algo:pib}, the previous lemma proves that the $\shifts$-pivot
degree of $\intBasis^{(2)} \intBasis^{(1)}$ is $\minDegs^{(1)} +
\minDegs^{(2)}$. In the rest of this section, we prove that the $\shifts$-Popov
form of $\intBasis^{(2)} \intBasis^{(1)}$ has the same $\shifts$-pivot degree
as $\intBasis^{(2)} \intBasis^{(1)}$. Consequently, the $\shifts$-minimal
degree of $(\evMat,\mulmat)$ is $\minDegs^{(1)} + \minDegs^{(2)}$ and thus can
be found from $\intBasis^{(2)}$ and $\intBasis^{(1)}$ without computing their
product.

It is known that left-unimodularly equivalent $\shifts$-reduced matrices have
the same $\shifts$-row degree up to permutation~\cite[Lemma~6.3-14]{Kailath80}.
Here, we prove that the $\shifts$-pivot degree is invariant among
left-unimodularly equivalent matrices in $\shifts$-weak Popov form.

\begin{lem}
\label{lem:invariant}
Let $\shifts \in \shiftSpace$ and let $\mat{P}$ and $\mat{Q}$ in $\intSpace$
be two left-unimodularly equivalent nonsingular polynomial matrices in
$\shifts$-weak Popov form. Then $\mat{P}$ and $\mat{Q}$ have the same
$\shifts$-pivot degree.
\end{lem}
\begin{proof}
Since row permutations preserve both the $\shifts$-pivot degrees and
left-unimodular equivalence, we can assume that $\mat{P}$ and $\mat{Q}$ are in
$\shifts$-diagonal weak Popov form. The $\shifts$-pivot degrees of $\mat{P}$
and $\mat{Q}$ are then $\rdeg[\shifts]{\mat{P}} - \shifts$ and
$\rdeg[\shifts]{\mat{Q}} - \shifts$, and it remains to check that
$\rdeg[\shifts]{\mat{P}} = \rdeg[\shifts]{\mat{Q}}$.

For any nonsingular $\mat{W} \in \polMatSpace$ in $\shifts$-weak Popov form we
have $\sumVec{\rdeg[\shifts]{\mat{W}}} = \deg(\det(\mat{W})) +
\sshifts$~\cite[Section~6.3.2]{Kailath80}. Thus, if $\mat{W}$ is furthermore
comprised entirely of rows in the $\polRing$-row space of $\mat{P}$ (that is,
$\mat{W}$ is a left multiple of $\mat{P}$) then we must have
$\sumVec{\rdeg[\shifts]{\mat{W}}} \ge \sumVec{\rdeg[\shifts]{\mat{P}}}$.

To arrive at a contradiction, suppose there exists a row index $i$ such that
the $\shifts$-degree of $\matrow{\mat{P}}{i}$ differs from that of
$\matrow{\mat{Q}}{i}$ and without loss of generality assume that the s-degree
of $\matrow{\mat{Q}}{i}$ is strictly less than that of $\matrow{\mat{P}}{i}$.
Then the matrix $\mat{W}$ obtained from $\mat{P}$ by replacing the $i$-th row
of $\mat{P}$ with $\matrow{\mat{Q}}{i}$ is in $\shifts$-diagonal weak Popov
form. This is a contradiction, since $\sumVec{\rdeg[\shifts]{\mat{W}}} <
\sumVec{\rdeg[\shifts]{\mat{P}}}$ and $\matrow{\mat{Q}}{i}$ is in the
$\polRing$-row space of $\mat{P}$ for $\mat{Q}$ is left-unimodularly
equivalent to $\mat{P}$.
\end{proof}

In particular, any nonsingular matrix in $\shifts$-weak Popov form has the same
$\shifts$-pivot degree as its $\shifts$-Popov form, which proves our point
about the $\shifts$-Popov form of $\intBasis^{(2)} \intBasis^{(1)}$.

\section{Computing interpolation bases with known minimal degree}
\label{sec:known_degrees}

In this section, we propose an efficient algorithm for computing the
$\shifts$-Popov interpolation basis $\intBasis$ for $(\evMat,\mulmat)$ when the
$\shifts$-minimal degree $\minDegs$ of $(\evMat,\mulmat)$ is known \emph{a
priori}.

First, we show that the shift $\shifts[d] = -\minDegs$ leads to the same
$\shifts[d]$-Popov interpolation basis $\intBasis$ as the initial shift
$\shifts$. Then, we prove that $\intBasis$ can be easily recovered from any interpolation
basis which is simply $\shifts[d]$-reduced. The following lemma
extends~\cite[Lemmas~15 and~17]{SarSto11} to the case of any shift~$\shifts$.

\begin{lem}
  \label{lem:minDeg_shift}
  Let $\shifts\in\shiftSpace$, and let $\mat{P}\in\polMatSpace[\rdim]$ be in
  $\shifts$-Popov form with column degree $\minDegs =
  (\minDeg_1,\ldots,\minDeg_\rdim)$. Then $\mat{P}$ is also in
  $\shifts[d]$-Popov form for $\shifts[d] =
  (-\minDeg_1,\ldots,-\minDeg_\rdim)$, and we have
  $\rdeg[{\shifts[d]}]{\mat{P}} = (0,\ldots,0)$. In particular, for any matrix
  $\mat{R} \in \polMatSpace[\rdim]$ which is unimodularly equivalent to
  $\mat{P}$ and $\shifts[d]$-reduced, $\mat{R}$ has column degree $\minDegs$,
  and $\mat{P} = \leadingMat[{\shifts[d]}]{\mat{R}}^{-1} \mat{R}$.
\end{lem}
\begin{proof}
  Let us denote $\mat{P} = [p_{ij}]_{i,j}$, and let $i \in \{1,\ldots,\rdim\}$.
  Since $\mat{P}$ is in $\shifts$-Popov form, it is enough to prove that the
  $\shifts[d]$-pivot entries of the rows of $\mat{P}$ are on its diagonal. We
  have $\deg(p_{ij}) < \deg(p_{jj}) = \minDeg_j$ for all $j\neq i$, and
  $\deg(p_{ii}) = \minDeg_i$. Then, the $i$-th row of $\mat{P}$ has
  $\shifts[d]$-pivot index $i$ and $\shifts[d]$-degree $0$. Thus $\mat{P}$ is
  in $\shifts[d]$-Popov form with $\shifts[d]$-row degree $(0,\ldots,0)$.

  Now, let $\mat{R}$ be a $\shifts[d]$-reduced matrix left-unimodularly
  equivalent to $\mat{P}$. Then, $\rdeg[{\shifts[d]}]{\mat{R}} =
  \rdeg[{\shifts[d]}]{\mat{P}} = (0,\ldots,0)$, so that we can write $\mat{R} =
  \leadingMat[{\shifts[d]}]{\mat{R}} \shiftMat{\minDegs} + \mat{Q}$ with the
  $j$-th column of $\mat{Q}$ of degree less than $\minDeg_j$. In particular,
  since $\leadingMat[{\shifts[d]}]{\mat{R}}$ is invertible, the column degree
  of $\mat{R}$ is $\minDegs$. Besides, we obtain
  $\leadingMat[{\shifts[d]}]{\mat{R}}^{-1} \mat{R} = \shiftMat{\minDegs} +
  \leadingMat[{\shifts[d]}]{\mat{R}}^{-1} \mat{Q}$, and the $j$-th column of
  $\leadingMat[{\shifts[d]}]{\mat{R}}^{-1} \mat{Q}$ has degree less than
  $\minDeg_j$. Thus $\leadingMat[{\shifts[d]}]{\mat{R}}^{-1} \mat{R}$ is in
  $\shifts[d]$-Popov form and unimodularly equivalent to $\mat{P}$, hence equal
  to $\mat{P}$.
\end{proof}

In particular, if $\minDegs$ is the $\shifts$-minimal degree of
$(\evMat,\mulmat)$ and $\shifts[d] = -\minDegs$, any $\shifts[d]$-minimal
interpolation basis $\mat{R}$ for $(\evMat,\mulmat)$ has size at most $\rdim^2
+ \rdim \sumVec{\minDegs}$, which for $\order \ge \rdim$ is in $\bigO{\rdim
\order}$. Still, the algorithm in~\cite{JeNeScVi15} cannot directly be used to
compute such an $\mat{R}$ efficiently, because $\sshifts[{\shifts[d] -
\min(\shifts[d])}]$ can be as large as $\Theta(\rdim \order)$, for example when
$\minDegs = (\order,0,\ldots,0)$; in this case, this algorithm uses
$\softO{\rdim^\expmatmul \order}$ operations.

By Lemma~\ref{lem:mindeg}, however, $\shifts[d] = -\minDegs$ satisfies
$\sshifts[{\!\max(\shifts[d]) - \shifts[d]}] \le \order$. For this type of
unbalanced shift, a solution in $\softO{\rdim^{\expmatmul-1} \order}$ already
exists in the particular case of order basis
computation~\cite[Section~6]{ZhoLab12}, building upon the partial linearization
technique in~\cite[Section~3]{Storjohann06}. Here, we adopt a similar approach,
taking advantage of the \emph{a priori} knowledge of the column degree of the
output matrix.

\begin{lem}
  \label{lem:known_mindeg_pib}
  Let $\evMat \in \matSpace[\rdim][\order]$, $\mulmat \in \matSpace[\order]$,
  and $\shifts \in \shiftSpace$, and let $\minDegs =
  (\minDeg_1,\ldots,\minDeg_\rdim)$ denote the $\shifts$-minimal degree of
  $(\evMat,\mulmat)$. 

  Then, let $\degExp = \lceil \order / \rdim \rceil \ge 1$, and for
  $i\in\{1,\ldots,\rdim\}$ write $\minDeg_i = (\quoExp_i -1) \degExp +
  \remExp_i$ with $\quoExp_i \ge 1$ and $0 \le \remExp_i < \degExp$, and let
  $\expand{\rdim} = \quoExp_1 + \cdots + \quoExp_\rdim$. Then, define
  $\expand{\minDegs}\in \NN^{\expand{\rdim}}$ as
  \begin{equation}
    \label{eqn:expandMinDegs}
    \expand{\minDegs} = ( \underbrace{\degExp, \ldots, \degExp, \remExp_1}_{\quoExp_1}, \ldots,
    \underbrace{\degExp, \ldots, \degExp, \remExp_\rdim}_{\quoExp_\rdim} )
  \end{equation}
  and the expansion-compression matrix $\expandMat \in
  \polMatSpace[\expand{\rdim}][\rdim]$ as
  \begin{equation}
    \label{eqn:expandMat}
    \expandMat = 
    \begin{bmatrix}
      1 \\
      X^\degExp \\
      \vdots \\
      X^{(\quoExp_1-1)\degExp} \\
      & \ddots \\
      & & 1 \\
      & & X^\degExp \\
      & & \vdots \\
      & & X^{(\quoExp_\rdim-1)\degExp}
    \end{bmatrix}. 
  \end{equation}
  Let further $\shifts[d] = - \expand{\minDegs} \in
  \shiftSpace[\expand{\rdim}]$ and $\mat{R} \in \polMatSpace[\expand{\rdim}]$
  be a $\shifts[d]$-minimal interpolation basis for $(\expandMat \mul \evMat,
  \mulmat)$. Then, the $\shifts$-Popov interpolation basis for
  $(\evMat,\mulmat)$ is the submatrix of
  $\leadingMat[{\shifts[d]}]{\mat{R}}^{-1} \mat{R} \expandMat$ formed by its
  rows at indices $\quoExp_1+\cdots+\quoExp_i$ for $1\le i\le \rdim$.
\end{lem}
\begin{proof}
  Let $\intBasis$ denote the $\shifts$-Popov interpolation basis for
  $(\evMat,\mulmat)$; $\intBasis$ has column degree $\minDegs$. First, we
  partially linearize the columns of $\intBasis$ in degree $\degExp$ to obtain
  $\widetilde{\intBasis} \in \polMatSpace[\rdim][\expand{\rdim}]$; more
  precisely, $\widetilde{\intBasis}$ is the unique matrix of degree less than
  $\degExp$ such that $\intBasis = \widetilde{\intBasis} \expandMat$.
  Then, we define $\expand{\intBasis} \in \polMatSpace[\expand{\rdim}]$ as
  follows:
  \begin{itemize}
    \item for $1\le i \le \rdim$, the row $\quoExp_1 + \cdots + \quoExp_i$ of $\expand{\mat{P}}$ is the
      row $i$ of $\widetilde{\intBasis}$;
    \item for $0\le i \le \rdim-1$ and $1 \le j \le \quoExp_{i+1}-1$, the row
      $\quoExp_1 + \cdots + \quoExp_i + j$ of $\expand{\intBasis}$ is the row
      $[ 0, \cdots, 0, X^\degExp, -1, 0, \cdots, 0] \in
      \polMatSpace[1][\expand{\rdim}]$ with the entry $X^\degExp$ at column
      index $\quoExp_1 + \cdots + \quoExp_i + j$.
  \end{itemize}
  Since $\intBasis$ is in $\shifts$-Popov form with column degree $\minDegs$,
  it is in $-\minDegs$-Popov form by Lemma~\ref{lem:minDeg_shift}. Then, one
  can check that $\expand{\intBasis}$ is in $\shifts[d]$-Popov form and has
  $\shifts[d]$-row degree $(0,\ldots,0)$.

  By construction, every row of $\expand{\intBasis}$ is an interpolant for
  $(\expandMat \mul \evMat,\mulmat)$. In particular, since $\mat{R}$ is an
  interpolation basis for $(\expandMat \mul \evMat,\mulmat)$, there is a matrix
  $\mat{U} \in \polMatSpace[\expand{\rdim}]$ such that $\expand{\intBasis} =
  \mat{U} \mat{R}$. Besides, there exists no interpolant $\expand{\row{p}} \in
  \polMatSpace[1][\expand{\rdim}]$ for $(\expandMat \mul \evMat,\mulmat)$ which
  has $\shifts[d]$-degree less than $0$: otherwise, $\expand{\row{p}}
  \expandMat$ would be an interpolant for $(\evMat,\mulmat)$, and it is easily
  checked that it would have $-\minDegs$-degree less than $0$, which is
  impossible.
  
  Thus every row of $\mat{R}$ has $\shifts[d]$-degree at least $0$, and the
  predictable degree property~\cite[Theorem~6.3.13]{Kailath80} shows that
  $\mat{U}$ is a constant matrix, and therefore unimodular. Then,
  $\expand{\intBasis}$ is an interpolation basis for $(\expandMat \mul
  \evMat,\mulmat)$, and since it is in $\shifts[d]$-Popov form, by
  Lemma~\ref{lem:minDeg_shift} we obtain that $\expand{\intBasis} =
  \leadingMat[{\shifts[d]}]{\mat{R}}^{-1} \mat{R}$. The conclusion follows.
\end{proof}

Then, it remains to prove that such a basis $\mat{R}$ can be computed
efficiently using the algorithm \algoname{MinimalInterpolationBasis}
in~\cite{JeNeScVi15}; this leads to Algorithm~\ref{algo:known_mindeg_pib}.

\begin{figure}[h!]
\centering
\fbox{\begin{minipage}{\linewidth}
  \begin{algo} \algoname{KnownMinDegMIB}
  \label{algo:known_mindeg_pib}
~\\
\emph{Input:} 
  \begin{itemize}
  \item a matrix $\evMat\in\evSpace{\order}$ with $\order \ge \rdim > 0$,
  \item a Jordan matrix $\jordan \in \matSpace[\order]$ in standard representation,
  \item a shift $\shifts\in\shiftSpace$,
  \item $\minDegs = (\minDeg_1,\ldots,\minDeg_\rdim) \in \NN^\rdim$ the
    $\shifts$-minimal degree of $(\evMat,\mulmat)$.
  \end{itemize}
\emph{Output:} the $\shifts$-Popov interpolation basis $\intBasis$ for
$(\evMat,\mulmat)$.
\begin{enumerate}[{\bf 1.}] 
  \item $\degExp \leftarrow \lceil \order / \rdim \rceil$, $\quoExp_i
    \leftarrow \lfloor \minDeg_i / \degExp \rfloor + 1$ for $1 \le i \le
    \rdim$, \\ $\expand{\rdim} \leftarrow \quoExp_1 + \cdots + \quoExp_\rdim$
  \item Let $\expand{\minDegs}\in \NN^{\expand{\rdim}}$ as in~\eqref{eqn:expandMinDegs}
    and $\shifts[d] \leftarrow - \expand{\minDegs} \in \NN^{\expand{\rdim}}$
  \item Let $\expandMat \in \polMatSpace[\expand{\rdim}][\rdim]$ as in~\eqref{eqn:expandMat}
    and $\expand{\evMat} \leftarrow \expandMat \mul \evMat$
  \item $\mat{R} \leftarrow \algoname{MinimalInterpolationBasis}(\expand{\evMat},\mulmat,\shifts[d] + (\degExp,\ldots,\degExp))$
  \item $\expand{\intBasis} \leftarrow \leadingMat[{\shifts[d]}]{\mat{R}}^{-1} \mat{R}$
  \item Return the submatrix of $\expand{\intBasis} \expandMat$
    formed by the rows at indices $\quoExp_1+\cdots+\quoExp_i$ for $1\le i\le \rdim$ 
\end{enumerate}
\end{algo}
\end{minipage}}
\end{figure}

\begin{prop}
  \label{prop:known_mindeg_pib}
  Assuming that $\mulmat\in\matSpace[\order]$ is a Jordan matrix given by a
  standard representation, and assuming we have the $\shifts$-minimal degree of
  $(\evMat,\mulmat)$ as an additional input, there is a deterministic algorithm
  \algoname{KnownMinDegMIB} which solves Problem~\ref{pbm:mib} using
  \begin{align*}
    \bigO{ & \rdim^{\expmatmul-1} \polmultime{\order} \log(\order) \log(\order/\rdim) }
    & \text{if } \expmatmul>2, \\
    \bigO{ & \rdim \polmultime{\order} \log(\order) \log(\order/\rdim) \log(\rdim)^3 }
    & \text{if } \expmatmul=2\;
  \end{align*}
  operations in $\field$.
\end{prop}
\begin{proof}
  We focus on the case $\order \ge \rdim$; otherwise, a better cost bound can
  be achieved even without knowing $\minDegs$~\cite[Theorem~1.4]{JeNeScVi15}.
  The correctness of Algorithm~\ref{algo:known_mindeg_pib} follows from
  Lemma~\ref{lem:known_mindeg_pib}. We remark that it uses $\shifts[d] +
  (\degExp,\ldots,\degExp)$ rather than $\shifts[d]$ because the minimal
  interpolation basis algorithm in~\cite{JeNeScVi15} requires the input shift
  to have non-negative entries. Since adding a constant to every entry of
  $\shifts[d]$ does not change the notion of $\shifts[d]$-reducedness, the
  basis $\mat{R}$ obtained at Step~\textbf{4} is a $\shifts[d]$-minimal
  interpolation basis for $(\evMat,\mulmat)$.

  Concerning the cost bound, we will show that it is dominated by the time
  spent in Step~\textbf{4}. First, we prove that
  $\sshifts[{\shifts[d]-\min(\shifts[d])}] \in \bigO{\order}$, so that the cost
  of Step~\textbf{4} follows from~\cite[Theorem~1.5]{JeNeScVi15}. We have
  $\quoExp_i = 1 + \lfloor \minDeg_i / \degExp \rfloor \le 1 + \rdim \minDeg_i
  / \order$ for all $i$. Thus, $\expand{\rdim} = \quoExp_1 + \cdots +
  \quoExp_\rdim \le \rdim + \sum_{1\le i\le \rdim} \rdim \minDeg_i / \order \le
  2\rdim$ thanks to Lemma~\ref{lem:mindeg}. Then, since all entries of
  $\shifts[d]$ are in $\{-\degExp,\ldots,0\}$, we obtain
  $\sshifts[{\shifts[d]-\min(\shifts[d])}] \le \expand{\rdim} \degExp \le 2
  \rdim (1 + \order/\rdim) \le 4 \order$. 

  Step~\textbf{3} can be done in~$\bigO{ \rdim \polmultime{\order}
\log(\order)}$ operations according to Lemma~\ref{lem:residual_expand} below.

  Lemma~\ref{lem:minDeg_shift} proves that the sum of the column degrees of
  $\mat{R}$ is $\sshifts[\expand{\minDegs}] = \sshifts[\minDegs] \le \order$.
  Then, the product in Step~\textbf{5} can be done in
  $\bigO{\rdim^{\expmatmul-1} \order}$ operations, by first linearizing the
  columns of $\mat{R}$ into a $\expand{\rdim} \times \expand{\rdim} +
  \sshifts[\expand{\minDegs}]$ matrix over $\field$, then left-multiplying this
  matrix by $\leadingMat[{\shifts[d]}]{\mat{R}}^{-1}$ (itself computed using
  $\bigO{\rdim^\expmatmul}$ operations), and finally performing the inverse
  linearization.

  Because of the degrees in $\expand{\intBasis}$ and the definition of
  $\expandMat$, the output in Step~\textbf{6} can be formed without using any
  arithmetic operation.
\end{proof}

The efficient computation of $\expandMat \mul \evMat$ can be done with the
algorithm for computing residuals in~\cite[Section~6]{JeNeScVi15}.

\begin{lem}
  \label{lem:residual_expand}
  The product $\expandMat \mul \evMat$ at Step~\textbf{3} of
  Algorithm~\ref{algo:known_mindeg_pib} can be computed using $\bigO{\rdim
    \polmultime{\order} \log(\order)}$ operations in $\field$.
\end{lem}
\begin{proof}
  The product $\expandMat \mul \evMat$ has $\expand{\rdim}$ rows, with
  $\expand{\rdim} \le 2\rdim$ as above. Besides, by definition of $\expandMat$,
  each row of $\expandMat \mul \evMat$ is a product of the form $X^{i \degExp}
  \mul \matrow{\evMat}{j}$, where $0 \le i \le \rdim$, $1 \le j \le \rdim$, and
  $\matrow{\evMat}{j}$ denotes the row $j$ of $\evMat$. In particular, $i
  \degExp \le 2\order$: then, according to~\cite[Proposition~6.1]{JeNeScVi15},
  each of these $\expand{\rdim}$ products can be performed using
  $\bigO{\polmultime{\order} \log(\order)}$ operations in $\field$.
\end{proof}

This lemma and the partial linearization technique can also be used to compute
the residual at Step~\textbf{2.c} of Algorithm~\ref{algo:pib}, that is, a
product of the form $\mat{P} \mul \evMat$ with the sum of the column degrees of
$\mat{P}$ bounded by $\order$. First, we expand the high-degree columns of
$\mat{P}$ to obtain $\expand{\mat{P}} \in \polMatSpace[\rdim][\expand{\rdim}]$
of degree less than $\lceil\order/\rdim\rceil$ such that $\mat{P} =
\expand{\mat{P}} \expandMat$; then, we compute $\expand{\evMat} = \expandMat
\mul \evMat$; and finally we rely on the algorithm
in~\cite[Proposition~6.1]{JeNeScVi15} to compute $\expand{\mat{P}} \mul
\expand{\evMat} = \mat{P} \mul \evMat$ efficiently.

\begin{cor}
  \label{cor:residual}
Let $\evMat \in \matSpace[\rdim][\order]$ with $\order \ge \rdim$, and let
$\mulmat \in \matSpace[\order]$ be a Jordan matrix given by a standard
representation. Let $\mat{P} \in \polMatSpace[\rdim]$ with column degree
$(\minDeg_1,\ldots,\minDeg_\rdim)$ such that $\minDeg_1 + \cdots +
\minDeg_\rdim \le \order$. Then, the product $\mat{P} \mul \evMat$ can be
computed using $\bigO{ \rdim^{\expmatmul-1} \polmultime{\order} \log(\order) }$
operations in $\field$.
\end{cor}

\noindent \textbf{Acknowledgments.} We thank B. Beckermann and G.
Labahn for their valuable comments, as well as an anonymous referee for suggesting a
shorter proof of Lemma~\ref{lem:invariant}. C.-P. Jeannerod and G. Villard
were partly supported by the ANR project HPAC (ANR 11 BS02 013). V.  Neiger was
supported by the international mobility grants Explo'ra Doc from \emph{R\'egion
Rh\^one-Alpes}, \emph{PALSE}, and \emph{Mitacs Globalink - Inria}. \'E.~Schost
was supported by NSERC.

\begin{small}

\end{small}

\appendix
\section{Reducing the entries of the shift}
\label{app:shift}

Let $\mat{A} \in \polMatSpace[\rdim]$ be nonsingular, let $\shifts \in
\shiftSpace$, and consider $\order \in \NN$ such that $\order >
\deg(\det(\mat{A}))$. Here, we show how to construct a shift $\shifts[t] \in
\NN^\rdim$ such that 
\begin{itemize}
  \item the $\shifts$-Popov form $\mat{P}$ of $\mat{A}$ is also in
    $\shifts[t]$-Popov form;
  \item $\min(\shifts[t]) = 0$, $\max(\shifts[t]) \le (\rdim-1) \order$,
    and $\sshifts[{\shifts[t]}] \le \rdim^2 \order / 2$.
\end{itemize}

We write $\shifts[\hat{s}] = (\shift{\pi(1)},\ldots,\shift{\pi(\rdim)})$ where
$\pi$ is a permutation of $\{1,\ldots,\rdim\}$ such that $\shifts[\hat{s}]$ is
non-decreasing. Then, we define $\shifts[\hat{t}] =
(\shift[\hat{t}]{1},\ldots,\shift[\hat{t}]{\rdim})$ by $\shift[\hat{t}]{1} = 0$
and, for $2 \le i \le \rdim$,
\[ \shift[\hat{t}]{i} - \shift[\hat{t}]{i-1} = \left\{ \begin{array}{ll}
  \order & \text{if } \shift[\hat{s}]{i} - \shift[\hat{s}]{i-1} \ge \order, \\
  \shift[\hat{s}]{i} - \shift[\hat{s}]{i-1} & \text{otherwise}.
  \end{array} \right.
\]
Let $\shifts[t] =
(\shift[\hat{t}]{\pi^{-1}(1)},\ldots,\shift[\hat{t}]{\pi^{-1}(\rdim)})$. Since
the diagonal entries of $\mat{P}$ have degree at most $\deg(\det(\mat{A})) <
\order$, we obtain that $\mat{P}$ is in $\shifts[t]$-diagonal weak Popov form
and thus in $\shifts[t]$-Popov form.

\section{Example of order basis with size beyond our target cost}
\label{app:large_size_mib}

We focus on a Hermite-Pad\'e approximation problem with input $\mat{F}$ of
dimensions $2\rdim \times 1$ as below, order $\order$ with $\order \ge \rdim$,
and shift $\shifts = (0,\ldots,0,\order,\ldots,\order) \in \NN^{2\rdim}$ with
$\rdim$ entries $0$ and $\rdim$ entries $\order$.

Let $f$ be a polynomial in $X$ with nonzero constant coefficient, and let
$f_1,\ldots,f_{\rdim}$ be generic polynomials in $X$ of degree less than
$\order$. Then, we consider the following input with all entries truncated
modulo $X^\order$:
\[\mat{F} =
  \trsp{[ f ,
    f + Xf ,
    X (f + X f) ,
    \cdots ,
    X^{\rdim-2} (f + X f) ,
    f_1 ,
    \cdots ,
  f_\rdim ]}.\]

After $\rdim$ steps, the iterative algorithm in~\cite{BecLab94} has computed an
$\shifts$-minimal basis $\intBasis^{(\rdim)}$ of approximants for $\mat{F}$ and
order~$\rdim$, which is such that $\shifts[t] =
\rdeg[\shifts]{\intBasis^{(\rdim)}} = (1,\ldots,1,\order,\ldots,\order)$ and
$ \intBasis^{(\rdim)} \mat{F} =
  \trsp{[  0 ,
    \cdots ,
    0 ,
    X^{\rdim} f ,
    X^{\rdim} g_1 ,
    \cdots ,
    X^{\rdim} g_\rdim
  ]} \bmod X^\order$,
for some polynomials $g_1,\ldots,g_\rdim$.

Now we finish the process up to order $\order$. Since the coefficient of degree
$\rdim$ of $X^\rdim f$ is nonzero and because of the specific shift
$\shifts[t]$, the obtained $\shifts$-minimal basis $\intBasis$ of approximants
for $\mat{F}$ has degree profile
\[
  \intBasis =
  \begin{bmatrix}
    [1] & [0]                                                 \\
    \vdots & \ddots & \ddots                         \\
    [1] & \cdots & [1] & [0]                            \\
    [d+1] & \cdots & [d+1] & [d+1] \\
    [d] & \cdots & [d] & [d] & [0]                      \\
    \vdots & \cdots & \vdots & \vdots & & & \ddots   \\
    [d] & \cdots & [d] & [d] & &   &     &     & [0]
  \end{bmatrix},
\]
where $d = \order-\rdim$, $[i]$ denotes an entry of degree $i$, the entries
left blank correspond to the zero polynomial, and the entries $[d+1]$ are on
the $\rdim$-th row. In particular, $\intBasis$ has size $\Theta(\rdim^2
\order)$.

\end{document}